\documentclass[aps,prx,twocolumn,superscriptaddress,floatfix,reprint]{revtex4}%{{{
\usepackage{amssymb}
\usepackage{amsthm}
\usepackage{mathtools}
\usepackage{dsfont}
\usepackage{graphicx}

\usepackage[draft,inline,nomargin]{fixme} \fxsetup{theme=color}
\newcommand{\tr}{{\rm tr}\,}
\newcommand{\ket}[1]{\left|{#1}\right\rangle}
\newcommand{\bra}[1]{\left\langle{#1}\right|}
\newcommand{\braket}[2]{\langle{#1}|{#2}\rangle}

\newcommand{\ketbra}[2]{\left|{#1}\rangle\!\langle{#2}\right|}

\DeclareMathOperator{\Tr}{Tr}

\newtheorem{theorem}{Theorem}
\newtheorem{lemma}{Lemma}

\begin{document}

\title{Quantum Stabilizer Channel for Thermalization}

\author{Esteban Mart\'inez Vargas}
%\affiliation{F\'isica Te\`orica: Informaci\'o i Fen\`omens Qu\`antics, Departament de F\'isica, Universitat Aut\`onoma de Barcelona, 08193 Bellatera (Barcelona) Spain}
\affiliation{Calle centella 3, Smza 18 Mza 3, 77505 Cancún México}
%\email{Esteban.Martinez@uab.cat}
\email{estebanmv@protonmail.com}
%\date{\today}

\begin{abstract}
    We study the problem of quantum thermalization via discrete interactions with 
    thermalized systems. We thus extend the previously introduced
    scattering thermalization program by studying not only a specific channel but 
    allowing any possible one. This problem is analogous to finding a channel with a
    desired fixed point.
    To solve this problem we give a complete characterization of
    any channel in terms of its fixed point. In theorem 1 we give through an SDP the completely positive 
    map that has the minimum trace which has the desired fixed point.
    The found linear map is not
    trace-preserving. In theorem 2 we complete
    the initial map so that it also becomes trace-preserving. We find that there is
    a family of channels that fulfill the desired fixed point condition.
    This characterization implies a strong separation of the channel in terms of 
    its Kraus operators: the first terms that project into the fixed point and other
    dissipative terms that fade away with several applications of the 
    channel.
    %From a quantum computing perspective, the
    %results thus obtained can be interpreted as a condition for quantum error correction that
    %also reminds of quantum error avoiding.
\end{abstract}

\maketitle
%}}}
\section{Introduction}%{{{
Thermalization is a ubiquitous phenomenon in the universe. Just as ubiquitous is the 
a~pplicability of quantum physics. This implies that a concept of quantum thermalization
should exist. It would mean that there must be a process where a physical system thermalizes
even when quantum physics is the most accurate description. This field of study
enters the wider subject of quantum thermodynamics \cite{strasberg2022quantum}.
It is however, not trivial to obtain a thermalized state from quantum 
dynamics: a thermal state is normally assumed as \emph{given} result.
%but
%the machinery that yields a thermal state is highly non-trivial.
%However, where is the limit of applicability of thermalization?
%When should the description as an aggregate of quantum particles be abandoned so that
%a collection of classical particles makes more sense? 

There has been work in the direction of obtaining thermal states as a result
of quantum dynamics, specifically, considering terminology
from quantum information theory 
\cite{EntanglementAnPopesc2006,ThermalizationRiera2012}. Their approach involves
having a Hamiltonian that dictates the evolution of a quantum state and 
observing when the system is close to a thermal state. Using general arguments of
typicality \cite{CoverThomasElements2006} one can find conditions where 
thermalization occurs and why assumptions of statistical mechanics arise naturally \cite{EntanglementAnPopesc2006}.
This line of thought can be called the program of dynamical typicality \cite{ThermalizationRiera2012}.

Not all systems fall into this scheme for thermalization.
A situation that is not strictly within the scope of the dynamical
typicality program 
is presented in the study of Jacob et. al. \cite{ThermalizationJacob2021,MicroReversibiEhrich2020} 
where thermalization in quantum systems is reached through a scattering process.
They consider a reservoir let us call it system $X$, that contains particles thermalized
to a specific temperature. This reservoir 
``shoots'' particles
spaced through time intervals into a central system, let us call it system $Y$,
that scatters them as shown in Fig. (\ref{fig:shtnbth}). 

\begin{figure}[h]
    \center
    \includegraphics[scale=2]{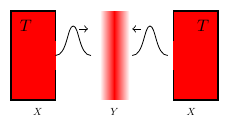}
    \caption{Two baths at temperature $T$ ``shooting'' particles to central system.}
    \label{fig:shtnbth}
\end{figure}

This approach differs from the dynamical typicality program in that
the interactions with the baths and the physical system in consideration do
not happen continuously. We have in change, a series of instantaneous
collisions or interactions that modify the system in consideration.
It could be called a ``quantum scattering thermalization program'' in juxtaposition to 
the aforementioned approach. 

However, Jacob et. al. approach failed to obtain a positive result in the quantum
limit. Their approach briefly explained in section \ref{sec:scthprog}, consists of
giving a Hamiltonian and observing how the system behaves, explicitly if it thermalizes or not.
One of their main results is to observe that it depends on the width of the wavepacket
that comes out from the thermal bath. The system only thermalizes for narrow wavepackets.
They also asked for microscopic reversibility.
In other words, they found that two conditions have to be met for thermalization
to occur with quantum particles in this scattering scheme:
\begin{itemize}
    \item Narrow wave packets.
    \item Microscopic reversibility.
\end{itemize}

Now, a question arises, why is condition 1 necessary? If one is presented with 
wide wave packets, as would require a 
natural quantum-mechanical description, thermalization should occur, as it is
a ubiquitous phenomenon in the universe. That is, thermalization should be 
independent from the wavelength. 

Notice however that the scattering thermalization program uses a specific channel
(see equation (\ref{eq:rhodinamics})) given by the scattering nature of this approach
and the specific Hamiltonian considered. The approach of Jacob et. al. can be seen as an investigation
of thermalization 
in discrete applications of this channel.

Here, we extend the scattering thermalization program by considering any 
possible channel. This is possible because in this problem we know the final state.
Thermalization has a very concrete meaning in this problem: that system $Y$
reaches a desired thermal state $\rho_Y$ given by the Hamiltonian of $Y$ and 
temperature $T$ from system $X$. Thus, we could fix the desired final state
and ask for the possible channels that yield it after many discrete iterations 
of such a channel. This is closely related to the concept of a fixed point of a
quantum channel. 

Suppose we have a trace-preserving channel $\Phi$, it has at least one fixed point $\sigma_f$ \cite{WatrousTheTheoryof2018}.
This is, there is at least one state that fulfills $\Phi[\sigma_f]=\sigma_f$.
If we start with any state $\rho$ the channel $\Phi$ will lead to $\sigma_f$ eventually,
$\Phi^n[\rho]\rightarrow\sigma_f$ for $n\rightarrow\infty$.

In this paper, we find the channel that yields the desired (thermal) state as a fixed point.
Being concrete, we give a family of channels that have a desired fixed point. 
We do so for channels of only one fixed point. Although more complex behavior is possible,
as there are channels with multiple fixed points and also rotating points \cite{AsymptoticsOfAlbert2019}. 
If one wants to thermalize to the state $\rho_Y$
one has to use a channel that has it as a fixed point. 

In this paper, we give two theorems in section \ref{sec:ItQCh}. In theorem \ref{thm:sdpteo} 
we show, using semidefinite programs \cite{WatrousTheTheoryof2018,ASemidefiniteEldar2003}
that the Choi matrix with minimum trace of the completely positive map that has a state $\sigma$
as a fixed point is given by
\begin{equation}
    Z_\sigma=\frac{\sigma\otimes(\ketbra{v_{max}}{v_{max}})^\intercal~}{\lambda_{max}},
\end{equation}
where $\lambda_{max}$ is $\sigma$'s largest eigenvalue and $\ket{v_{max}}$ its correspondent
eigenvector. Observe that
\begin{equation}
    \Phi[\sigma]=\tr_{\mathcal{H}_2}[Z_\sigma(\mathds{1}_{\mathcal{H}_1}\otimes \sigma^{\intercal})]=\sigma.
\end{equation}
This is in general, not a trace preserving map. To recover trace preservation we
complement the characterization in theorem \ref{thm:decomposition}, where we show that
we can ``complete'' the channel $Z_\sigma$ into a trace-preserving map by adding
an additional part $B$,
\begin{align}
    \mathcal{C}[\sigma,B]&=\sigma\otimes\frac{(\ketbra{v_{max}}{v_{max}})^\intercal}{\lambda_{max}}\nonumber\\
    &+B\otimes(\mathds{I}-\frac{(\ketbra{v_{max}}{v_{max}})^\intercal}{\lambda_{max}}).
\end{align}
This is a complete characterization of channels in terms of their fixed point. 
This characterization implies a strong separation in this type of channel, there is
a conservative part, that leaves the fixed point intact and there is a dissipative part
which decays with several applications of the channel. In the Choi representation as
well as in the Kraus representation, we observe this strong separation.

We build immediately in Sec. \ref{sec:kraus} the Kraus operators for the found channel
family and observe that they separate into two parts: the fixed point part and the $B$ which
completes the channel.
Afterward in section \ref{sec:JacobEx} we state how to reconcile Jacob et. al.'s approach with the one we
propose here. We also state the relationship with quantum error correction in section
\ref{sec:RelQEC}. We end with a discussion in section \ref{sec:disc}.
%}}}
\section{Quantum scattering thermalization program}%{{{
\label{sec:scthprog}
We now describe the scattering thermalization program. 
We have a Hamiltonian that encodes the evolution of the
independent subsystems and the interaction between them. Given by
\begin{equation}
    H=H_0+V,
\end{equation}
where $H_0=p^2/2m\otimes\mathds{I}_Y+\mathds{I}_Y\otimes H_Y$
is the sum of the kinetic energy of $X$ and the internal energy
of the system $Y$.
We can therefore define the unitaries $U_0(t)=\text{exp}[-itH_0/\hbar]$ and 
$U(t)=\text{exp}[-itH/\hbar]$. We can define the isometric M\o ller operators~\cite{taylor2006scattering}
\begin{equation}
    \Omega_{\pm}=\lim_{t\rightarrow\mp\infty}U^\dagger(t)U_0(t).
\end{equation}
We define the scattering operator as
\begin{equation}
    S=\Omega^\dagger_{-}\Omega_{+},
\end{equation}
which is unitary, $SS^\dagger=S^\dagger S=\mathds{I}$.
Defining the initial state of both $X$ and $Y$ as $\rho_X\otimes\rho_Y$
we have the following map
\begin{equation}
    \mathds{S}[\rho_Y] = \Tr_{X}[S(\rho_X\otimes\rho_Y)S^\dagger].
    \label{eq:rhodinamics}
\end{equation}
Once the interaction $V$ is specified, the Hamiltonian $H$ implies a unitary evolution with the
interaction and then without it. 
If we define the map $\varepsilon_t(\cdot)=\text{exp}[-itH_Y/\hbar](\cdot)\text{exp}[itH_Y/\hbar]$
then we can write the dynamics at step $n$ as
\begin{equation}
    \rho_Y^{(n)}=\varepsilon_{\tau_n}\circ\mathds{S}\circ\ldots\varepsilon_{\tau_2}\circ\mathds{S}\circ\varepsilon_{\tau_1}\circ\mathds{S}[\rho_Y^{(0)}],
\end{equation}
where $\circ$ denotes channel composition.
From the analysis of a scattering map that acts on mixed states we 
can divide the wave packets into those that are broad in the momentum variance and
those who are narrow.
One of the main results from Ref. \cite{ThermalizationJacob2021} 
is that a necessary condition for thermalization in the scattering scenario
is that the wave packets have to be narrow.
Another central condition (but not so relevant in our case) is microscopic reversibility, 
which is the reason that two baths are shooting into a central system.
%}}}
\section{Iterations of a quantum channel}%{{{
\label{sec:ItQCh}
In our problem we know the state the source is producing, what we don't know is the initial target state.
We ask for a channel that through iterations on a system would change it to a desired state 
regardless of the initial one. The mathematical problem can be stated explicitly:
find a quantum channel that produces the desired thermal state starting
from another given state after many iterations, in other words, find a
quantum channel with a desired fixed point.
This topic is thus related to the stabilizer formalism used to prove the Knill-Gottesman theorem
\cite{NielsenChuangQuantum2011,ImprovedSimulaAarons2004}. 

Suppose we are given a thermal state that depends on a Hamiltonian $H$ as $\rho_{th}[H]$, it is
then relevant to ask if there exists a (nontrivial) channel $\Phi$ such that it is the \emph{stabilizer
    channel} of $\rho_{th}[H]$ in the sense that
\begin{equation}
    \Phi(\rho_{th}[H])=\rho_{th}[H].
    \label{eq:ChanProy}
\end{equation}

The Choi representation of channels \cite{WatrousTheTheoryof2018} will be useful for addressing the existence
of a nontrivial channel that fulfills equation (\ref{eq:ChanProy}).
We start with a completely positive map
$\Phi$ that fulfills equation (\ref{eq:ChanProy}), obtained
as a solution to the semidefinite program (SDP) \cite{boyd2004convex,WatrousTheTheoryof2018,ASemidefiniteEldar2003}

\begin{equation}
    \begin{aligned}
        & \underset{Z}{\text{minimize}}
        & & \Tr[Z] \\
        & \text{subject to}
        & &  \tr_{\mathcal{H}_2}[Z(\mathds{1}_{\mathcal{H}_1}\otimes \sigma^{\intercal})]\geq\sigma\\
        & & & Z\geq 0.
    \end{aligned}
    \label{eq:sdp}
\end{equation}
Observe that the solution of the SDP $Z$ corresponds to the Choi matrix of the map 
$\Phi$. Also, observe that the constraint $Z\geq 0$ assures us that we are looking for
a completely positive map. The minimization
of the trace of $Z$ is relevant, as it eliminates any (unnecessary) 
orthogonal element. This means, suppose that $K=R+Z$ fulfills the conditions of the
SDP and $Z$ is the optimal solution, the minimization assures that 
\begin{equation}
    \tr_{\mathcal{H}_2}[R(\mathds{1}_{\mathcal{H}_1}\otimes \rho_{th}[H]^{\intercal})]=0
\end{equation}
for $R\geq0$.
If $\sigma=\rho_{th}[H]$ in the SDP (\ref{eq:sdp}) then we have solved our problem
at hand. Fortunately, this SDP can be explicitly solved as we show in the
following theorem.
\begin{theorem}
    \label{thm:sdpteo}
    The SDP (\ref{eq:sdp}) has the solution $1/\lambda_{max}$ where $\lambda_{max}$ is
    the maximum eigenvalue of $\sigma$.
\end{theorem}
\begin{proof}
    We build the dual program to (\ref{eq:sdp}),
\begin{equation}
    \begin{aligned}
        & \underset{W}{\text{maximize}}
        & & \Tr[W\sigma] \\
        & \text{subject to}
        & &  W^\intercal\otimes\sigma^\intercal\leq\mathds{1}\\
        & & & W\geq 0.
    \end{aligned}
    \label{eq:sdpdual}
\end{equation}
Let 
\begin{equation}
    \sigma=\sum_i\lambda_i\ketbra{v_i}{v_i}
\end{equation}
be the spectral decomposition of $\sigma$.
We define
\begin{equation}
    Z_\sigma:=\frac{\sigma\otimes(\ketbra{v_{max}}{v_{max}})^\intercal~}{\lambda_{max}},\quad\quad W_\sigma:=\frac{1}{\lambda_{max}}\mathds{1}.
\end{equation}
where $\ket{v_{max}}$ is the correspondent eigenvector for $\lambda_{max}$.
Observe that $Z_\sigma$ and $W_\sigma$ belong in the primal and dual feasible sets of their respective programs. They
also yield the same value for the figure of merit, so strong duality is always fulfilled. The value is thus
the optimal one because of Slater's theorem for semidefinite programs~\cite{WatrousTheTheoryof2018}. 
\end{proof}

The optimal solution map $Z_\sigma$ is not trace-preserving,
observe that an arbitrary state $\rho$ would result in
\begin{equation}
    \Phi(\rho)= \frac{\bra{v_{max}}\rho\ket{v_{max}}~}{\lambda_{max}}\sigma.
    \label{eq:chanlambdamax}
\end{equation}
The map is trace-preserving only when 
\begin{equation}
\bra{v_{max}}\rho\ket{v_{max}}=\lambda_{max}.
\label{eq:imprel}
\end{equation}
Also, notice that the solution $Z_\sigma$ fulfills the condition of
the stabilizer channel exactly, not only solving the SDP (\ref{eq:sdp}).

However, we can study the cases where the equality (\ref{eq:imprel}) is fulfilled. 
Observe that this happens when $\rho=U_{v_{max}}DU_{v_{max}}^\dagger$ and 
\begin{equation}
    D = 
\begin{pmatrix}
    \lambda_{max} & 0 \\
    0 & (1-\lambda_{max})\Lambda_{d-1} 
\end{pmatrix}
\label{eq:conds}
\end{equation}
$\lambda_{max}\geq1/2$ w.l.g. and $\Lambda_{d-1}$ is a $(d-1)\times(d-1)$ positive diagonal matrix with $\tr[\Lambda_{d-1}]=1$.
$U_{v_{max}}$ is a unitary matrix of dimension $d$ which includes $\ket{v_{max}}$
as a column that corresponds to $\lambda_{max}$. These states denote the set of states 
that are stabilized by the map (\ref{eq:chanlambdamax}) without loss.

For the case of one qubit the freedom that equation (\ref{eq:chanlambdamax})
allows is null: specifying one eigenvalue specifies the second one,
also specifying a vector specifies its orthogonal one (on the opposite side
of Bloch's sphere). Imagine now that we have two qubits, the effective dimension
of the Hilbert space is four. The freedom here is much more interesting,
we can specify the eigenvalue matrix following equation (\ref{eq:conds}) and
the $U_{max}$ can be specified as follows
\begin{equation}
    U_{max} = 
\begin{pmatrix}
    1 & 0 & 0 & 0 \\
    0 & v^1_1 & v^2_1 & v^3_1 \\
    0 & v^1_2 & v^2_2 & v^3_2 \\
    0 & v^1_3 & v^2_3 & v^3_3 
\end{pmatrix},
\end{equation}
where $v^i_j$ represents the $j$th entry of the $i$th eigenvector.
Generalizing, $n$ qubits imply that there is a Hilbert space of dimension $2n-1$
where we can choose freely a state that represents the allowed errors.

To consider channels we need the constraint of trace-preservation.
Remember that the SDP (\ref{eq:sdp}) required the minimum possible expression 
for the map in question. 
We have to append a remaining part of the map so that it
becomes a trace-preserving one and thus a quantum channel. This process provides a complete characterization
of channels as we show in the following theorem. However, first we need to show a lemma.
\begin{lemma}
    \label{lma:ineqs}
    \begin{align}
        \sigma&\otimes\frac{(\ketbra{v_{max}}{v_{max}})^\intercal}{\lambda_{max}}+B\otimes(\mathds{I}-\frac{(\ketbra{v_{max}}{v_{max}})^\intercal}{\lambda_{max}})\geq0\nonumber\\
        &\Longleftrightarrow\quad\quad B\geq0,\nonumber\\
        &\quad\quad\quad\sigma\geq(1-\lambda_{max})B.
    \end{align}
\end{lemma}
\begin{proof}
    $\Rightarrow$ First we multiply by a matrix and apply a partial trace
    \begin{align}
        \tr_X&[(\mathds{I}\otimes(\ketbra{v_{max}}{v_{max}})^\intercal)(\sigma\otimes\frac{(\ketbra{v_{max}}{v_{max}})^\intercal}{\lambda_{max}})\\&
        +B\otimes(\mathds{I}-\frac{(\ketbra{v_{max}}{v_{max}})^\intercal}{\lambda_{max}}))]\geq0,\nonumber\\
        \frac{\sigma}{\lambda_{max}}&+B(1-\frac{1}{\lambda_{max}})\geq0,\\
        \sigma&\geq(1-\lambda_{max})B.
    \end{align}
    Doing the same with any other eigenvector of $\sigma$, i.e. for $v_i\neq v_{max}$
    we have
    \begin{align}
        \tr_X&[(\mathds{I}\otimes(\ketbra{v_i}{v_i})^\intercal)(\sigma\otimes\frac{(\ketbra{v_{max}}{v_{max}})^\intercal}{\lambda_{max}}\\&
        +B\otimes(\mathds{I}-\frac{(\ketbra{v_{max}}{v_{max}})^\intercal}{\lambda_{max}}))]\geq0,\nonumber\\
        B\geq0.
    \end{align}
    $\Leftarrow$ Observe now that as the tensor product preserves the positivity,
    then
    \begin{align}
        \frac{1}{\lambda_{max}}(\sigma+(\lambda_{max}-1)B)\otimes(\ketbra{v_{max}}{v_{max}})^\intercal&\geq0\nonumber\\
        B\otimes\mathds{I}_{/v_{max}}\geq0.
        \label{eq:elemsum}
    \end{align}
    where $\mathds{I}_{/v_{max}}$ denotes the sum of all the eigenvector
    basis elements minus the one which correspond to $v_{max}$. Summing both
    inequalities of Eq. (\ref{eq:elemsum}) we get
    \begin{equation}
        \sigma\otimes\frac{(\ketbra{v_{max}}{v_{max}})^\intercal}{\lambda_{max}}+B\otimes(\mathds{I}-\frac{(\ketbra{v_{max}}{v_{max}})^\intercal}{\lambda_{max}}))]\geq0.
    \end{equation}
\end{proof}
We can now state and prove the main theorem.
\begin{theorem}
    \label{thm:decomposition}
    Given a state $\sigma$ with maximum eigenvalue $\lambda_{max}$ and correspondent 
    eigenvector $\ket{v_{max}}$, 
    we can describe a family of channels with fixed point $\sigma$
    in terms of its Choi matrix $\mathcal{C}$ as follows
    \begin{align}
        \mathcal{C}[\sigma,B]&=\sigma\otimes\frac{(\ketbra{v_{max}}{v_{max}})^\intercal}{\lambda_{max}}\nonumber\\
        &+B\otimes(\mathds{I}-\frac{(\ketbra{v_{max}}{v_{max}})^\intercal}{\lambda_{max}}).
        \label{eq:chandecomp}
    \end{align}
    %$B$ is a state that fulfills $\bra{v_{max}}B\ket{v_{max}}\leq\lambda_{max}$.
    $B$ is a state that fulfills $\sigma\geq(1-\lambda_{max})B$.
\end{theorem}

Equation (\ref{eq:chandecomp}) implies that there is a relevant positive semidefinite part of
a channel is given by the operator
\begin{equation}
    P_\sigma:= \sigma\otimes\frac{(\ketbra{v_{max}}{v_{max}})^\intercal}{\lambda_{max}}.
\end{equation}

\begin{proof}
    Because of theorem \ref{thm:sdpteo}
    the operator with minimum trace with fixed point $\sigma$ has to be $P_\sigma$. 
    In the case that $\bra{v_{max}}\rho\ket{v_{max}}<\lambda_{max}$ the map $P_\sigma$
    would be non-trace-preserving. To have trace preservation we need to add an additional
    operator.
    We consider an operator $B^\prime$ such that the Choi matrix of $\Phi$ is
    \begin{equation}
        \mathcal{C}=P_\sigma+B^\prime.
    \end{equation}
    The fixed point condition $\Phi[\sigma]=\sigma$ requires that
    \begin{equation}
        \tr_{\mathcal{H}_2}[B^\prime(\mathds{I}\otimes\sigma^\intercal)]=0.
    \end{equation}
    For trace preservation, we have that for any state $\rho$,
    \begin{equation}
        \tr[B^\prime(\mathds{I}\otimes\rho^\intercal)]=1-\frac{\bra{v_{max}}\rho\ket{v_{max}}~}{\lambda_{max}}.
    \end{equation}
    By theorem 2.26 of \cite{WatrousTheTheoryof2018}, for trace-preserving maps
    \begin{equation}
        \tr_{\mathcal{H}_1}[C]=\mathds{I}.
    \end{equation}
    This is achieved by an operator of the form
    \begin{equation}
        B^\prime=B\otimes(\mathds{I}-\frac{\ketbra{v_{max}}{v_{max}~}}{\lambda_{max}})
        \label{eq:Bdecomp}
    \end{equation}
    with $\tr[B]=1$. 
    %This defines a family of channels, explicitly, those that can be written as
    %in Eq. (\ref{eq:chandecomp}).

    Also, as $B\geq0$ and $\sigma\geq(1-\lambda_{max})B$
    by lemma (\ref{lma:ineqs}) we have that the right part of Eq. (\ref{eq:chandecomp}) is
    positive semidefinite. Therefore Eq. (\ref{eq:chandecomp}) describes a completely positive
    map and thus a quantum channel.
\end{proof}
    Observe that in general there is
    an infinite number of channels (one for each possible $B$) corresponding to a single fixed point.
    The resulting operator yields the channel acting one and two
    times on an operator $\rho$ as follows
    \begin{align}
        \Phi[\rho]&=\sigma\frac{\bra{v_{max}}\rho\ket{v_{max}~}}{\lambda_{max}}+\nonumber\\
        &B(1-\frac{\bra{v_{max}}\rho\ket{v_{max}~}}{\lambda_{max}}),\nonumber\\
        \Phi^2[\rho]&=\sigma\frac{\bra{v_{max}}\rho\ket{v_{max}~}}{\lambda_{max}}\\
        &+\sigma\frac{\bra{v_{max}}B\ket{v_{max}~}}{\lambda_{max}}(1-\frac{\bra{v_{max}}\rho\ket{v_{max}~}}{\lambda_{max}})\nonumber\\
        &+B(1-\frac{\bra{v_{max}}B\ket{v_{max}~}}{\lambda_{max}})\times\nonumber\\
        &(1-\frac{\bra{v_{max}}\rho\ket{v_{max}~}}{\lambda_{max}}).
    \end{align}
    We observe that the proportion of $\sigma$ grows with each iteration and the proportion of $B$ decreases.
    As we asked for quantum channels then $\Phi^n[\rho]\rightarrow\sigma$ for $n\rightarrow\infty$.
    
    \begin{figure}
        \includegraphics[width=0.5\textwidth]{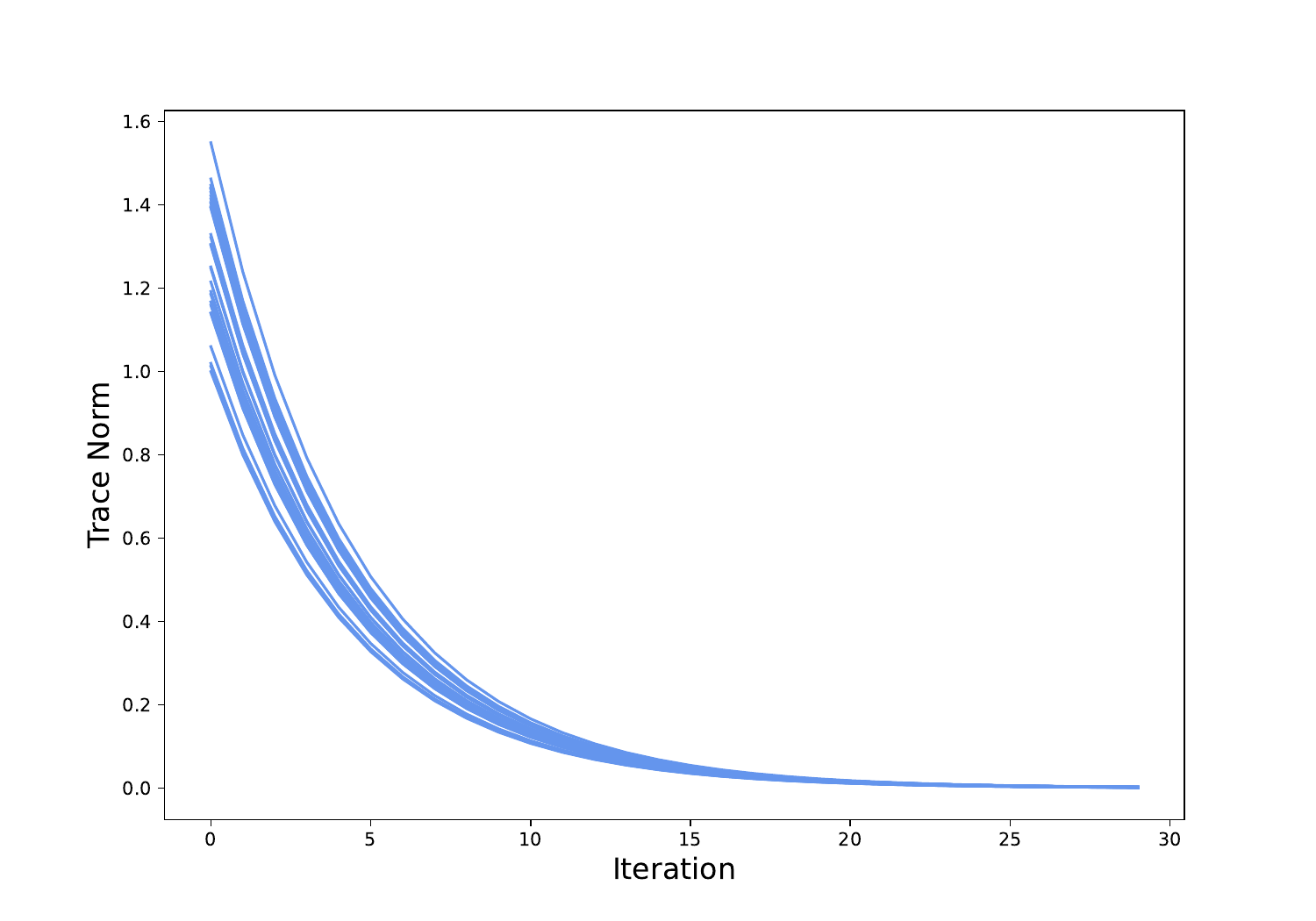}
        \caption{We plot the trace norm between the iterated channel 
            and the fixed point as a function of iterations of the channel, this 
            means $n$ in the horizontal axis and $||\Phi^n[\rho]-\sigma||_1$ in the
            vertical axis.
            This is done for states of 
            dimension 5. The figure shows 30 samples of $\sigma$ produced randomly. The fixed point
            is pure.}
        \label{fig:TraceDist}
    \end{figure}
    
    In Fig. (\ref{fig:TraceDist}) we observe how the trace norm between the nth iteration of the 
    channel and its fixed point decays with iterations, this is, $||\Phi^n[\rho]-\sigma||_1$ where
    $\sigma$ is $\Phi$'s fixed point. In the figure we use pure fixed states for the fixed points. 
    To produce Fig. (\ref{fig:TraceDist}) a pure state $\sigma$ of dimension 
    5 is chosen randomly with the Haar measure. Then, two states $\rho$ and $B$ are chosen randomly so that they
    fulfill the conditions of theorem \ref{thm:decomposition} (we sample random states 
    until we get a pair that fulfills those conditions).
    Finally, the trace-preserving channel is built following Eq. (\ref{eq:chandecomp}) and applied several times 
    starting with the input state $\rho$. Fig. (\ref{fig:TraceDist}) portrays 30 samplings
    of initial $\sigma$.
    We observe that the trace norm between the iterated state and the
    fixed point decays rapidly.

    We observe a possible interpretation of the action of the channel.
    The first part $P_{\sigma}$ projects the input state to the fixed point $\sigma$ and the
    second part dissipates. The second part of the Choi matrix yields zero when applied
    to $\sigma$, so it never deviates. The part that is projected to $B$ then is constantly
    projected to the fixed point. This part is dissipative.

    For example, suppose we $\sigma$ is a thermal state diagonal in the eigenvector base 
    of a Hamiltonian $H$ then, 
    \begin{equation}
        \Phi[\rho] =\sigma\frac{\bra{v_{max}}\rho\ket{v_{max}~}}{\lambda_{max}}+B(1-\frac{\bra{v_{max}}\rho\ket{v_{max}~}}{\lambda_{max}}),
    \end{equation}
    is non-diagonal as we have possibly, coherences from the state $B$. However, after
    several applications of $\Phi$ the $B$ part will decay and only the diagonal terms
    will survive.

%}}}
    \section{Kraus operators}%{{{ 
    \label{sec:kraus}
    We can define the operators $A_{i}$ as
    \begin{equation}
        A_k:=\sqrt{\frac{\lambda_k}{\lambda_{max}}~}\ketbra{v_k}{v_{max}}.
        \label{eq:KrausAs}
    \end{equation}
%    Observe that $A_{i}$ has two sub-indices and can therefore be ordered.
%    We thus define the operator 
%    \begin{equation}
%        A := 
%    \begin{pmatrix}
%        A_1 & 0 &\ldots & 0 \\
%        0 & A_2 & \ldots & 0 \\
%        \vdots & \vdots  & \ddots & \vdots\\
%        0 & 0 & \ldots & A_d
%    \end{pmatrix}.
%    \end{equation}
%    This operator acts on the state with an ancilla with a Hilbert space
%    $\mathcal{Z}$ of dimension $d$,
%    \begin{equation}
%        \Phi(\sigma)=\tr_{\mathcal{Z}}[A(\mathds{1}_{d}\otimes \sigma) A^\dagger].
%    \end{equation}
%    In an analogous fashion of Eq. (\ref{eq:KrausAs}) we can define the following 
    Analogously, we define the Kraus
    operators for the complementary part of the channel. 
    Let us decompose $B$ as
    \begin{equation}
        B = \sum_{k}b_k\ketbra{b_k}{b_k},
    \end{equation}
    then we can define
    \begin{align}
        B_{kl} &:=\sqrt{b_k}\ketbra{b_k}{v_l}\nonumber\\
        B_{k, max} &:= \sqrt{b_k\left(\frac{\lambda_{max}-1}{\lambda_{max}}\right)}\ketbra{b_{k}}{v_{max}}.
        \label{eq:KrausB}
    \end{align}

    We notice that these operators would be orthogonal to the fixed point $\sigma$. Also,
    the operators $A_i$ commute with the fixed point. Observe that for the fixed point
        \begin{equation}
            \sum_{l}B_l\sigma B^{\dagger}_l=\tr[\sigma\left(\mathds{I}-\frac{\ketbra{v_{max}}{v_{max}}~}{\lambda_{max}}\right)]=0.
        \end{equation}

        The operators described in Eqs. (\ref{eq:KrausAs}) and (\ref{eq:KrausB}) define a 
        completely positive map $\Phi[\rho]$ for $\bra{v_{max}}\rho\ket{v_{max}}\leq\lambda_{max}$. 
        Also notice that
%        trace-preserving
%        transformation, as we have that
        \begin{equation}
            \sum_kA^\dagger_kA_k = \sum_k\frac{\lambda_k}{\lambda_{max}}\ketbra{v_{max}}{v_{max}}=\frac{\ketbra{v_{max}}{v_{max}}~}{\lambda_{max}}.
        \end{equation}
        On the other hand,
        \begin{align}
            \sum_{kl}B^\dagger_{kl}B_{kl} &=\sum_{kl}b_k\ketbra{v_l}{v_l}-\frac{\ketbra{v_{max}}{v_{max}}~}{\lambda_{max}}\\
            &=\mathds{I}-\frac{\ketbra{v_{max}}{v_{max}}~}{\lambda_{max}}.
        \end{align}
        Therefore, as $\sum_kA^\dagger_kA_k+\sum_{kl}B^\dagger_{kl}B_{kl}=\mathds{I}$ the operators form a 
        trace-preserving map \cite{WatrousTheTheoryof2018}. 
%    This means that we have proven an extension
%    of a theorem due to Kribs \cite{Kribs_2003} (theorem 4.25 from \cite{WatrousTheTheoryof2018}).
%    \begin{theorem}
%        Let $\{A_k,B_l\}$  be a collection of operators satisfying
%        \begin{equation}
%            \Phi[X] = \sum_{k}A_kXA^{\dagger}_k + \sum_{l}B_lXB^{\dagger}_l,
%        \end{equation}
%        for all states $X$. For every state $X$ it holds that $\Phi[X]=X$ if
%        and only if $[X,A_k]=0$ for every $k$ and $\sum_lB_lXB^{\dagger}_l=0$.
%    \end{theorem}
        %}}}
        \section{Jacob et. al. example}%{{{
\label{sec:JacobEx}
    For example, in Jacob et. al.'s problem, they consider a qubit state as the central one, remember from Fig. (\ref{fig:shtnbth}).
    The structure of the problem is a quantum channel over a central system $\rho_{H_2}$ given
    in equation (\ref{eq:rhodinamics})
    \begin{equation}
        \Phi_T[\rho_Y]:=\tr_X[S(\rho_X\otimes\rho_Y)S^\dagger],
    \end{equation}
    with $S$ a unitary matrix that contains the interactions and the Hamiltonian.
    With the use of a state $\rho_X$ each time the channel acts on $\rho_Y$. 
    Expanding the state of the first Hilbert space like
    \begin{equation}
        \rho_X=\sum_ir_j\ketbra{r_j}{r_j},
    \end{equation}
    we get a Kraus representation of the quantum channel,
    \begin{equation}
        \Phi_T[\rho_Y]=\sum_{ij}r_j^{1/2}\bra{r_i}S\ket{r_j}\rho_Yr_j^{1/2}\bra{r_j}S^\dagger\ket{r_i}.
    \end{equation}
    This channel has the correspondent Choi matrix \cite{WatrousTheTheoryof2018},
    \begin{equation}
        J(\Phi)=\sum_{ij}\text{Vec}(r_j^{1/2}\bra{r_i}S\ket{r_j})\text{Vec}(r_j^{1/2}\bra{r_j}S^\dagger\ket{r_i})^*.
        \label{eq:ChoiWtrs}
    \end{equation}
    Their main analysis is to observe that the quantum channel they are building depends
    on the state $\rho_X$. Basically there are two possible channels, each one has a
    different fixed point. Therefore one channel has a fixed point that correspond to 
    the thermalized state and the other to other state. Using the theory developed here
    we know how to build channels that always thermalize.

    Notice that if Kraus operators have efectively the form described in Sec. (\ref{sec:kraus}) then
    we will get a thermalizing channel. This is irrespective of the global form of the $S$ operators.
    In other words, the global operator could correspond to a wide wavepacket and 
    thermalization would still occur.

%    Suppose that this system has a Hamiltonian $H$. The thermal state associated with this state is
%    \begin{equation}
%        \rho_c=\frac{e^{-\beta H}}{\mathcal{Z}}.
%    \end{equation}
%    Observe then that theorem \ref{thm:decomposition} restricts the unitary operations $S$
%    that thermalize to this state into those with Choi matrix,
%    \begin{equation}
%        J(\Phi_T) = \mathcal{C}[\rho_c,B], 
%        \label{eq:JeqsC}
%    \end{equation}
%    as defined by Eqs. (\ref{eq:ChoiWtrs}) and (\ref{eq:chandecomp}).
    We can observe how the unitary transformations
    $S$ should behave in terms of the fixed point $\sigma$ and the complement $B$
    using the Kraus operators calculated in Sec. \ref{sec:kraus}. Notice that
    %\fxnote{CHECK INDICES!!}
    \begin{equation}
        r_l^{1/2}\bra{r_k}S\ket{r_l} = \sqrt{\frac{\lambda_i}{\lambda_{max}}~}\ketbra{v_i}{v_{max}},
    \end{equation}
    for suitable $j$'s. Therefore, we have that 
    \begin{equation}
        r_l^{1/2}\bra{v_i}\bra{r_k}S\ket{r_l}\ket{v_{max}} = \sqrt{\frac{\lambda_i}{\lambda_{max}}~}.
    \end{equation}
    For the correspondent Kraus operators of the $B$ complement part (Eq. (\ref{eq:KrausB}))
    we have
    \begin{equation}
        r_l^{1/2}\bra{v_i}\bra{r_k}S\ket{r_l}\ket{v_j} = \sqrt{b_n}\braket{v_i}{b_n}.
    \end{equation}
    We have a term with the following value
    \begin{equation}
        r_k^{1/2}\bra{v_i}\bra{r_k}S\ket{r_l}\ket{v_{max}} = \sqrt{b_n\left(\frac{\lambda_{max}-1}{\lambda_{max}}\right)}\braket{v_i}{b_n}.
    \end{equation}
    This gives a physical interpretation of how to solve Jacob et. al.'s problem: the Kraus operators
    have to have the form of those given in Sec. \ref{sec:kraus}. In other terms, for the
    family of channels described by Eq. (\ref{eq:chandecomp}) there is a strong separation
    of terms, some are conservative denoted by $A_k$ and some are dissipative denoted by
    $B_{kl}$. We observe that there is always a strong dissipative part which eliminates 
    coherences from whatever the initial state is and leaves the fixed point state.

    Notice, however, a main difference between the approach of Jacob et. al. and the one
    taken in this paper is that they are dealing with continous variables whereas we
    deal with finite dimensional quantum systems. This is why they do not build the
    Kraus operators of their map and we are able to do so.

%    We can show numerically that using a channel written in the form of Eq. (\ref{eq:chandecomp})
%    a state \emph{always} thermalize. The problem of Jacob et. al. is that they
%    use two different channels with different fixed points. We can build random
%    channels that have all the same fixed point. The numerics are made for the qubit case.
%    In Fig. \ref{fig:BlochPaths} we plot states in the Bloch sphere. We start with pure
%    states on the surface, these states are chosen randomly with the Haar measure. 
%    Then, we apply the channel $\mathcal{C}[\mathds{1}/2,B]$ with $B$ chosen randomly
%    with the Haar measure. We therefore apply to random input states a random quantum
%    channel. We apply the channel to each initial state $n=40$ times. We see there is
%    \emph{always} a convergence towards the chosen fixed point: $\mathds{1}/2$.
%
%    \begin{figure}
%        \includegraphics[width=0.49\textwidth]{BlochRandPaths-crop.pdf}
%        \caption{Random channels with random starting points in the
%            Bloch sphere always converges to the fixed point. In this graph, we use $n=40$ 
%            iterations. However, the number of iterations needed to reach the fixed
%            point $\rho=\mathds{1}/2$ is usually much lower. The graph shows 10 examples
%            of starting states and channels.}
%        \label{fig:BlochPaths}
%    \end{figure}
%
%}}}
\section{The relationship with Quantum Error Correction}%{{{
\label{sec:RelQEC}
The initial condition (\ref{eq:ChanProy}) is the same used in quantum error 
correction \cite{QuantumErrorCTerhal2015}. Nevertheless, this approach is for mixed states 
meanwhile most approaches use pure states.
The construction of a stabilizer channel for mixed states has been already discussed, in
theorem 4.8 of \cite{WatrousTheTheoryof2018}, which states that a 
stabilizer channel that uses information that leaks into the environment to correct the state in 
question can be built
as a mixed-unitary channel. However, we present here the same problem from a different angle, as
our approach is less ambitious.
which allows us to get to the relation (\ref{eq:imprel}) for error correctability.

There are two approaches to deal with errors for quantum 
technologies: first, correct them using quantum error correction
codes (QECC), and second, avoid errors using quantum error avoiding codes
(QEAC). QECC is related to error-correcting codes for pure states, its approach
is to construct the stabilizer operator for a state. On the other hand, QEAC
is more focused on mixed states and finding strategies so that the quantum
information is preserved in decoherence-free subspaces. Our approach seems something
in the middle: it uses mixed states and a channel that presupposes
openness to an environment and asks for the preservation of a subsystem, however,
it also is based on a stabilizer operator. Notice however, that our requirement is
very mild, we only ask for the conservation of the eigenvector correspondent
to the maximum eigenvalue and the maximum eigenvalue as well. However, we would
have to extend (or restrict) the conditions to allow the final state to be
not a specific one, as we do here but a corrected state.

Theorems \ref{thm:sdpteo} and \ref{thm:decomposition} are for mixed states and as
several pure states can yield the same mixed state, the error correction protocols
based on these theorems have to take this into account.
%}}}
\section{Discussion}%{{{
\label{sec:disc}

We address the problem of thermalization by studying a \emph{stabilizer channel} that, 
given any initial state of dimension $d<\infty$ 
it yields the desired thermal state after many actions on a physical system. 
Our approach is to characterize channels in terms of their fixed points. Explicitly,
for channels with one fixed point. We find in theorem \ref{thm:sdpteo}
the completely positive map with the least trace that has a given state $\sigma$ as its fixed point.
In theorem \ref{thm:decomposition} we complete the map into a quantum channel given in the
previous theorem by adding trace-preservation.
The Kraus operators for the channel built with our method are obtained.
%and we observe
%that we extend a theorem of fixed points in quantum channels 
%due to Kribs \cite{Kribs_2003,WatrousTheTheoryof2018}.
We then observe the conditions
to reconcile our approach with Jacob et. al.'s discussion based on a specific channel and 
we made numerical calculations with 5 dimensional qudits.
Finally, we discuss a possible connection with quantum error correction.

The most relevant aspect of the results presented here, compared to those of
Jacob et. al. is that they restrict to a specific unitary transformation whereas we allow
any channel. We find channels that thermalize for any
initial state, thus recovering the ubiquitous nature of thermalization.
However, it poses the possibility of requiring dynamics that are not 
intuitive, that is, not only Schrödinger equation but open-system dynamics.

Due to the original motivation of this research, a question of interpretation
arises. Jacob et. al. \cite{ThermalizationJacob2021} proposed a ``natural'' channel for thermalizing
a system towards a given state: the scattering process. Given the negative answer to this 
natural assumption, we ask if there is another channel that does thermalize. We find there is
and arrive at the conclusion that the Kraus operators have to coincide with the ones
given in Sec. \ref{sec:kraus}. This yields a physical interpretation of the channels
we study, they confer a specific form that the Kraus operator of a quantum channel must
have to have a desired fixed point. In other words, the channels described here
separate into two parts: a conservative and a dissipative part. The conservative 
preserves the fixed point subspace. The dissipative part eliminates all the rest.

Our approach yields at least two questions to be answered.
First of all, there is the question of resources to build the channel. 
Therefore, how to quantify the resources to build the channel that we want, i.e.
the one that corresponds to $Z_\sigma$? This is a
question that demands the use of tools and terminology from resource 
theories \cite{UsingAndReusiDiaz2018,ResourceTheoryAlbare2018}. Also, to characterize the possible states $B$ to complete
the channel in theorem (\ref{thm:decomposition}), so that it becomes
a trace-preserving channel.

Secondly, there is a relationship between the channel in question and quantum error correction. Specifically,
it is related to the stabilizer formalism. One of the great
arguments against quantum computation is the difficulty of taming the errors from the
environment \cite{1908.02499v1}. Further work using the formalism developed here should
address these arguments. Specifically, to have stabilizer channels that correct
errors concerning desired symmetries.

Finally, there is plenty of room for generalization. In general, a quantum channel can have
several fixed points and rotating points. The family of quantum channels studied here 
have a specific structure and have only a fixed point.
%}}}
%\section*{Acknowledgements}%{{{
%I thank fruitful discussions with P. Strasberg on the topic. Also discussions
%with J. M. R. Parrondo and M. Esposito.%}}}
%\include{Appendix}
\bibliography{bibliography.bib}
\end{document}